\documentclass[10pt]{extarticle}

\usepackage[english]{babel}
\usepackage[ansinew]{inputenc}
\usepackage{amsmath}
\usepackage{amsfonts}
\usepackage{amssymb}
\usepackage{amsthm}
\usepackage{graphicx}
\usepackage{fancybox}
\usepackage{calc}
\usepackage{tabularx}
\usepackage{tikz}
\usepackage{array}
\usepackage{lmodern}
\usepackage{mathabx}
\usepackage{fancybox}
\usepackage{caption}
\usepackage{fullpage}

\usepackage[margin=1.6in]{geometry}

\newcommand\numberthis{\addtocounter{equation}{1}\tag{\theequation}}
\newcommand{\E}{\mathbb{E}}

\newtheorem{Thm}{Theorem}
\newtheorem{Lem}[Thm]{Lemma}

\newtheorem{Cor}{Corollary}

\newtheorem{Def}{Definition}

\newtheorem{Exm}{Example}

\newtheorem{Rem}{Remark}

\newtheorem{Ctn}{Construction}

\title{On the hardness of the Learning with Errors problem with a discrete reproducible error distribution
\thanks{The research was supported by the DFG Research Training Group GRK $1817/1$}
}
\author{
Filipp Valovich\\ Horst G\"{o}rtz Institute for IT Security\\ Faculty of Mathematics\\ Ruhr-Universit\"{a}t Bochum, Universit\"{a}tsstra{\ss}e 150, 44801 Bochum, Germany\\ Email: filipp.valovich@rub.de
}
\date{}

\allowdisplaybreaks

\begin{document}

\maketitle

\pagenumbering{gobble}
\pagestyle{empty}

\begin{abstract}
In this work we show that the hardness of the Learning with Errors problem with errors taken from the discrete Gaussian distribution implies the hardness of the Learning with Errors problem with errors taken from the symmetric Skellam distribution. Due to the sample preserving search-to-decision reduction by Micciancio and Mol the same result applies to the decisional version of the problem. Thus, we provide a variant of the Learning with Errors problem that is hard based on conjecturally hard lattice problems and uses a discrete error distribution that is similar to the continuous Gaussian distribution in that it is closed under convolution. As an application of this result we construct a post-quantum cryptographic protocol for differentially private data anlysis in the distributed model. The security of this protocol is based on the hardness of the new variant of the Decisional Learning with Errors problem. A feature of this protocol is the use of the same noise for security and for differential privacy resulting in an efficiency boost.
\end{abstract}

\section{Introduction}

In recent years the \textit{Learning with Errors} (LWE) problem received a lot of attention in the cryptographic research community. As instance of the LWE problem we are given a uniformly distributed matrix $\textbf{\upshape A}\in\mathbb{Z}_q^{\lambda\times\kappa}$ and a noisy codeword $\textbf{\upshape y}=\textbf{\upshape Ax}+\textbf{\upshape e}\in\mathbb{Z}_q^{\lambda}$ with an error term $\textbf{\upshape e}\in\mathbb{Z}_q^{\lambda}$ chosen from a proper error distribution $\chi^{\lambda}$ and an unknown $\textbf{\upshape x}\in\mathbb{Z}_q^{\kappa}$. The task is to find the correct vector $\textbf{\upshape x}$. In the decisional version of this problem (DLWE problem) we are given $(\textbf{\upshape A}, \textbf{\upshape y})$ and have to decide whether $\textbf{\upshape y}=\textbf{\upshape Ax}+\textbf{\upshape e}$ or $\textbf{\upshape y}$ is a uniformly distributed vector in $\mathbb{Z}_q^{\lambda}$. In \cite{42} a search-to-decision reduction was provided to show that the two problems are essentially equivalent in the worst case and in \cite{43} a sample preserving search-to-decision reduction was provided for certain cases showing the equivalence in the average case. Moreover, in \cite{42} the average-case-hardness of the search problem was established by the construction of an efficient quantum algorithm for worst-case-lattice-problems using an efficient solver of the LWE problem if the error distribution $\chi$ is a \textit{discrete Gaussian} distribution. Accordingly, most cryptographic applications of the LWE problem used a discrete Gaussian error distribution for their constructions.\\
In this work we are concerned with the question whether the hardness of the LWE problem can be established for other discrete distributions, especially for reproducible distributions, i.e. distributions that are closed under convolution (- the discrete Gaussian distribution is not reproducible). This question is motivated by the following application. In \cite{2} the notion of \textit{Private Stream Aggregation} (PSA) was introduced. A PSA scheme is a cryptographic protocol between $n$ users and an untrusted aggregator. It enables each user to securely send encrypted time-series data to the aggregator. The aggregator is then able to decrypt the aggregate of all data in each time step, but cannot retrieve any further information about the individual data. In \cite{40} it was shown that a PSA scheme can be built upon any key-homomorphic weak pseudo-random function and some security guarantees were provided. In this paper we instantiate a concrete PSA scheme with a key-homomorphic weak pseudo-random function constructed from the DLWE problem.\\
A PSA scheme enables the users to output statistics over their data under \textit{differential privacy}. This notion was introduced in \cite{8} and is a measure for statistical disclosure of private data. Usually, $(\epsilon,\delta)$-differential privacy is preserved by using a mechanism that adds properly distributed noise to the statistics computed over a database with individual private data of users. In most of the works on differentialy privacy, these mechanisms are considered in the \textit{centralised} setting, where a trusted authority collects the individual data in the clear and performs the perturbation process. In contrast, a PSA scheme allows the users to perform differentially private data analysis in the \textit{distributed setting}, i.e. without the need of relying on a trusted authority. In light of that, the \textit{Skellam mechanism} was introduced in \cite{40} and shown to preserve differential privacy. The advantage of the Skellam mechanism over other mechanisms (like the Laplace \cite{8}, the Exponential \cite{9}, the Geometric \cite{52}, the Gaussian \cite{14} or the Binomial \cite{14}) mechanisms is that it is discrete (enabling cryptographic operations), maintains relatively high accuracy and is reproducible. This property allows all users to generate noise of small variance, that sums up to the value for the required level of differential privacy. Therefore the Skellam mechanism is well-suited for an execution through a PSA scheme.\\
We will take advantage of these properties of the Skellam distribution for our DLWE-based PSA scheme by using errors following the symmetric Skellam distribution $\text{\upshape Sk}_\mu$ with variance $\mu$ rather than the discrete Gaussian distribution. Therefore we need to show the average-case-hardness of the LWE problem with errors taken from the Skellam distribution. Now we can state the main theorem that will be shown in this work.

\begin{Thm}\label{lweskthm}
Let $\kappa$ be a security parameter 
and let $\lambda=\lambda(\kappa)=\text{\upshape\sffamily poly}(\kappa)$ with $\lambda>3\kappa$. Let $q=q(\kappa)=\text{\upshape\sffamily poly}(\kappa)$ be a sufficiently large prime modulus and $\rho>0$ such that $\rho q\geq 2\lambda\sqrt{\kappa}$. If there exists a PPT-algorithm that solves the $\text{\upshape LWE}(\kappa,\lambda,q,{\text{\upshape Sk}_{(\rho q)^2/4}})$ problem with non-negligible probability, then there exists an efficient quantum-algorithm that approximates the decision-version of the shortest vector problem ($\text{\upshape GAPSVP}$) and the shortest independent vectors problem ($\text{\upshape SIVP}$) to within $\tilde{O}(\lambda\kappa/\rho)$ in the worst case.
\end{Thm}

Based on the same assumptions, the decisional version $\text{\upshape DLWE}(\kappa,\lambda,q,{\text{\upshape Sk}_{(\rho q)^2/4}})$ of the problem is also hard due to the search-to-decision reduction in \cite{43}. Hence, the error terms in our DLWE-based PSA scheme are used for two tasks: establishing the cryptographic security of the scheme and the distributed generation of noise for preserving differential privacy.\\
Our proof of Theorem \ref{lweskthm} is inspired by techniques used in \cite{39} where a variant of the LWE problem with uniform errors on a small support is shown to be hard. As in \cite{39}, we will construct a \textit{lossy code} for the error distribution $\text{\upshape Sk}_\mu$ from the LWE problem with discrete Gaussian errors. Variants of lossy codes were first used in \cite{53} and since then had applications in different hardness reductions, such as the reduction from the LWE problem to the Learning with Rounding problem in \cite{54}. Lossy codes are pseudo-random codes seeming to be good codes. However, encoding messages with a lossy code and adding certain errors annihilates any information on the message. On the other hand, encoding the same message using a truly random code and adding the same type of error preserves the message. We will conclude that recovering the message when encoding with a random code and adding noise must be computationally hard. If this was not the case, lossy codes could be efficiently distinguished from random codes, contradicting the pseudo-randomness-property of lossy codes. As in \cite{39} and opposed to the LWE problem with discrete Gaussian errors, our worst-to-average case reduction depends on the number $\lambda$ of LWE-samples. Thus, we will consider a $\lambda$-bounded LWE problem, where $\lambda$ has a fixed $\text{\upshape\sffamily poly}(\kappa)$ upper bound. This does not restrict our application to PSA and differential privacy, since we will identify $\lambda$ with the total number of queries processed during the execution of PSA.\\
Essential definitions and known facts about used distributions, the Learning with Errors problem and lossy codes are given in Sections \ref{nad} and \ref{fac} respectively. The proof of Theorem \ref{lweskthm} is provided in Section \ref{hardn} and in Section \ref{applic} we show how this result can be applied to the construction of a PSA scheme, thus yielding a prospective post-quantum cryptographic protocol for data analyses under differential privacy.

\begin{Rem} The result from $\cite{39}$ does not provide a proper error distribution for our DLWE-based PSA scheme since a differentially private mechanism with uniform noise provides no accuracy to statistical data analyses.
\end{Rem}

\begin{Rem} The original result from $\cite{42}$ states that the LWE problem is hard in the set $\mathbb{T}=\mathbb{R}/\mathbb{Z}$ when the noise is distributed according to the \mbox{\em continuous} Gaussian distribution (with a certain bound on the variance) modulo $1$. Although the continuous Gaussian distribution is reproducible as well, it does not seem to be a good idea to instantiate it instead of the Skellam distribution in the DLWE-based PSA scheme: For data processing reasons the values would have to be discretised. The resulting noise would follow a discrete Gaussian distribution which is not reproducible any more. The larger the number of users the more independent samples of discrete Gaussian noise would have to be added together. Moreover, if the total noise has to be invariant to the number of users (e.g. in order to keep the accuracy of the statistics), then the discretisation causes a stronger deviation from the discrete Guassian distribution the larger the number of users becomes. Therefore the deviation from a truely $(\epsilon,\delta)$-differentially private analysis would scale with the number of users.
\end{Rem}

\section{Preliminaries}\label{nad}

Let $q>2$ be a prime. We handle elements from $\mathbb{Z}_q$ as their central residue-class representation. This means that $x^\prime\in\mathbb{Z}_q$ is identified with $x\equiv x^\prime$ mod $q$ for $x\in\{-(q-1)/2,\ldots,(q-1)/2\}$ thereby lifting $x^\prime$ from $\mathbb{Z}_q$ to $\mathbb{Z}$.

\subsection{Distributions}

Let $X$ be a set. We denote by $\mathcal{U}(X)$ the uniform distribution on $X$. Let $\chi$ be a distribution (on $X$). We denote by $x\leftarrow\chi$ (or sometimes $x\leftarrow\chi(X)$) the sampling of $x$ (from $X$) according to $\chi$. If $A\leftarrow\chi^{a\times b}$ (or $A\leftarrow\chi(X^{a\times b})$) then $A$ is an $a\times b$-matrix constructed by picking every entry independently (from $X$) according to the distribution $\chi$.

\begin{Def}[Discrete Gaussian Distribution\cite{42}]
Let $q$ be an integer and let $\Phi_s$ denote the normal distribution with variance $s^2/(2\pi)$. Let $\overline{\Psi}_\alpha$ denote the discretised Gaussian distribution with variance $(\alpha q)^2/(2\pi)$, i.e. $\overline{\Psi}_\alpha$ is sampled by taking a sample from $\Phi_{\alpha q}$ and performing a randomised rounding $\cite{44}$. Let $D_\nu$ be the discretised Gaussian distribution with variance $\nu$, i.e. $D_\nu=\overline{\Psi}_{\sqrt{2\pi\nu}/q}$.
\end{Def}

\begin{Def}[Skellam Distribution \cite{29}]\label{skellam} Let $\mu_1$, $\mu_2> 0$. A discrete random variable $X$ is drawn according to the Skellam distribution with parameters $\mu_1,\mu_2$ (short: $X\leftarrow \text{\upshape Sk}_{\mu_1,\mu_2}$) if it has the following probability distribution function $\psi_{\mu_1,\mu_2}\colon\mathbb{Z}\mapsto\mathbb{R}$:
\[\psi_{\mu_1,\mu_2}(k)=\exp(-(\mu_1+\mu_2))\left(\frac{\mu_1}{\mu_2}\right)^{k/2}I_k(2\sqrt{\mu_1\mu_2}),\]
where $I_k$ is the modified Bessel function of the first kind (see pages $374$--$378$ in $\cite{28}$).
\end{Def}

A random variable $X\leftarrow \text{\upshape Sk}_{\mu_1,\mu_2}$ has variance $\mu_1+\mu_2$ and can be generated as the difference of two random variables drawn according to the Poisson distribution of mean $\mu_1$ and $\mu_2$, respectively \cite{29}. Note that the Skellam distribution is not generally symmetric. However, we only consider the particular case $\mu_1=\mu_2=\mu/2$ and refer to this symmetric distribution as $\text{Sk}_\mu = \text{Sk}_{\mu/2,\mu/2}$.

\subsection{Learning with Errors}\label{lwesubsec}

We will consider a $\lambda$-bounded LWE problem, where the adversary is given $\lambda(\kappa) = \text{\upshape\sffamily poly}(\kappa)$ samples (which we can write conveniently in matrix-form). As observed in \cite{39}, this consideration poses no restrictions to most cryptographic applications of the LWE problem, since they require only an a-priori fixed number of samples. In our application to differential privacy (Section \ref{applic}) we will identify $\lambda$ with the number of queries in a pre-defined time-series.\\

\noindent\textit{Problem} $1$. $\boldsymbol\lambda$\textbf{-bounded LWE Search-Problem, Average-Case Version.} Let $\kappa$ be a security parameter, let $\lambda = \lambda(\kappa) = \text{\upshape\sffamily poly}(\kappa)$ and $q = q(\kappa)\geq 2$ be integers and let $\chi$ be a distribution on $\mathbb{Z}_q$. Let $\textbf{x}\leftarrow\mathcal{U}(\mathbb{Z}_q^\kappa)$, let $\textbf{A}\leftarrow\mathcal{U}(\mathbb{Z}_q^{\lambda\times\kappa})$ and let $\textbf{e}\leftarrow\chi^\lambda$. The goal of the $\text{\upshape LWE}(\kappa,\lambda,q,\chi)$ problem is, given $(\textbf{A}, \textbf{Ax} + \textbf{e})$, to find $\textbf{x}$.\\

\noindent\textit{Problem} $2$. $\boldsymbol\lambda$\textbf{-bounded LWE Distinguishing-Problem.} Let $\kappa$ be a security parameter, let $\lambda = \lambda(\kappa) = \text{\upshape\sffamily poly}(\kappa)$ and $q = q(\kappa)\geq 2$ be integers and let $\chi$ be a distribution on $\mathbb{Z}_q$. Let $\textbf{x}\leftarrow\mathcal{U}(\mathbb{Z}_q^\kappa)$, let $\textbf{A}\leftarrow\mathcal{U}(\mathbb{Z}_q^{\lambda\times\kappa})$ and let $\textbf{e}\leftarrow\chi^\lambda$. The goal of the $\text{\upshape DLWE}(\kappa,\lambda,q,\chi)$ problem is, given $(\textbf{A}, \textbf{y})$, to decide whether $\textbf{y}=\textbf{Ax} + \textbf{e}$ or $\textbf{y}=\textbf{u}$ with $\textbf{u}\leftarrow\mathcal{U}(\mathbb{Z}_q^\lambda)$.\\

\subsection{Entropy and Lossy Codes}

We introduce the map-conditional entropy as starting point for our technical tools. It can be seen as a measure for ambiguity.
\begin{Def}[Map-conditional entropy]
Let $\chi$ be a probability distribution with finite support and let $X\leftarrow\chi$. Let $Supp(\chi)$ be the support of $\chi$, $\xi\in Supp(\chi)$ and let $f,g$ be two (possibly randomised) maps on the domain $Supp(\chi)$. The $(f,g,\xi)$-conditional entropy $H_{f,g,\xi} (X)$ of $X$ is defined as 
\[H_{f,g,\xi}(X)=-\log_2(\Pr[X=\xi\,|\,f(X)=g(\xi)]).\]
\end{Def}

In the remainder of the paper we will see $f=f_{\textbf{A},\textbf{e}}$ and $g=g_{\textbf{A},\textbf{e}}$ as maps to the set of LWE instances, i.e.
\[f_{\textbf{A},\textbf{e}}(\textbf{y})=g_{\textbf{A},\textbf{e}}(\textbf{y})=\textbf{Ay} + \textbf{e}.\]
In this work we consider the $(f_{\textbf{A},\textbf{e}},f_{\textbf{A},\tilde{\textbf{e}}},\tilde{\textbf{x}})$-conditional entropy
\[H_{f_{\textbf{A},\textbf{e}},f_{\textbf{A},\tilde{\textbf{e}}},\tilde{\textbf{x}}}(\textbf{x})=-\log_2(\Pr[\textbf{x}=\tilde{\textbf{x}}\,|\,\textbf{Ax} + \textbf{e}=\textbf{A}\tilde{\textbf{x}} + \tilde{\textbf{e}}]),\]
of a random variable $\textbf{x}$, i.e. the entropy of $\textbf{\upshape x}$ given that a LWE instance generated with $(\textbf{\upshape A},\textbf{\upshape x},\textbf{e})$ is equal to another LWE instance generated with $(\textbf{\upshape A},\tilde{\textbf{\upshape x}},\tilde{\textbf{e}})$. Now we provide the notion of lossy codes, which is a main technical tool used in the proof of the hardness result.

\begin{Def}[Families of Lossy Codes \cite{39}]\label{lossydef}
Let $\kappa$ be a security parameter, let $\lambda = \lambda(\kappa) = \text{\upshape\sffamily poly}(\kappa)$ and let $q=q(\kappa)\geq 2$ be a modulus, $\Delta=\Delta(\kappa)$ and let $\chi$ be a distribution on $\mathbb{Z}_q$. Let $\{\mathcal{C}_{\kappa,\lambda,q}\}$ be a family of distributions where $\mathcal{C}_{\kappa,\lambda,q}$ is defined on $\mathbb{Z}_q^{\lambda\times\kappa}$. The distribution family $\{\mathcal{C}_{\kappa,\lambda,q}\}$ is $\Delta$-lossy for the error distribution $\chi$, if the following hold:
\begin{enumerate}
 \item $\mathcal{C}_{\kappa,\lambda,q}$ is pseudo-random: It holds that $\mathcal{C}_{\kappa,\lambda,q}\approx_c\mathcal{U}(\mathbb{Z}_q^{\lambda\times\kappa})$.
 \item $\mathcal{C}_{\kappa,\lambda,q}$ is lossy: Let $f_{\textbf{\upshape B},\textbf{\upshape b}}(\textbf{\upshape y})=\textbf{\upshape B}\textbf{\upshape y} + \textbf{\upshape b}$. Let $\textbf{\upshape A}\leftarrow\mathcal{C}_{\kappa,\lambda,q}$, $\tilde{\textbf{\upshape x}}\leftarrow\mathcal{U}(\mathbb{Z}_q^\kappa), \tilde{\textbf{\upshape e}}\leftarrow\chi^\lambda$, let $\textbf{\upshape x}\leftarrow\mathcal{U}(\mathbb{Z}_q^\kappa)$ and $\textbf{\upshape e}\leftarrow\chi^\lambda$. Then it holds that 
 \[\Pr_{(\textbf{\upshape A},\tilde{\textbf{\upshape x}},\tilde{\textbf{\upshape e}})}[H_{f_{\textbf{\upshape A},\textbf{\upshape e}},f_{\textbf{\upshape A},\tilde{\textbf{\upshape e}}},\tilde{\textbf{\upshape x}}}(\textbf{\upshape x})\geq\Delta]\geq 1-\text{\upshape\sffamily neg}(\kappa).\]
 \item $\mathcal{U}(\mathbb{Z}_q^{\lambda\times\kappa})$ is non-lossy: Let $f_{\textbf{\upshape B},\textbf{\upshape b}}(\textbf{\upshape y})=\textbf{\upshape B}\textbf{\upshape y} + \textbf{\upshape b}$. Let $\textbf{\upshape A}\leftarrow\mathcal{U}(\mathbb{Z}_q^{\lambda\times\kappa})$, $\tilde{\textbf{\upshape x}}\leftarrow\mathcal{U}(\mathbb{Z}_q^\kappa), \tilde{\textbf{\upshape e}}\leftarrow\chi^\lambda$, let $\textbf{\upshape x}\leftarrow\mathcal{U}(\mathbb{Z}_q^\kappa)$ and $\textbf{\upshape e}\leftarrow\chi^\lambda$. Then it holds that 
\[\Pr_{(\textbf{\upshape A},\tilde{\textbf{\upshape x}},\tilde{\textbf{\upshape e}})}[H_{f_{\textbf{\upshape A},\textbf{\upshape e}},f_{\textbf{\upshape A},\tilde{\textbf{\upshape e}}},\tilde{\textbf{\upshape x}}}(\textbf{\upshape x})=0]\geq 1-\text{\upshape\sffamily neg}(\kappa).\]
\end{enumerate}
\end{Def}

Note that our definition of lossy codes deviates from the definition of lossy codes provided in \cite{39}, since we use another type of entropy (which in general may be larger than the conditional min-entropy considered in \cite{39}). We will see that our notion of map-conditional entropy suffices for showing the hardness of the LWE problem.

\section{Basic Facts}\label{fac}

\subsection{Facts about the used Distributions}

We need to find a value such that a random variable distributed according to the discrete Gaussian distribution exceeds this value only with negligible probability.

\begin{Lem}[Bound on the discrete Gaussian distribution]\label{disgausbound}
Let $\kappa$ be a security parameter, let $s=s(\kappa)=\omega(\log(\kappa))$ and let $\nu=\nu(\kappa)=\text{\upshape\sffamily poly}(\kappa)$. Let $\zeta=\zeta(\kappa)>0$ be an integer. Let $g_1,\ldots,g_\zeta\leftarrow D_\nu$. Then \[\Pr\left[\left|g_1+\ldots +g_\zeta\right|>\sqrt{\zeta\nu s}\right]\leq\text{\upshape\sffamily neg}(\kappa).\]
\end{Lem}
\begin{proof}
Since $g_1,\ldots,g_\zeta$ are independent sub-gaussian random variables with sub-gaussian parameter $\sqrt{\nu}$, the result follows from an application of the Hoeffding-type inequality for sub-gaussian random variables (see Proposition $5.10$ in \cite{51}).
\end{proof}

We use the fact that the sum of independent Skellam random variables is a Skellam random variable, i.e. the Skellam distribution is reproducible.

\begin{Lem}[Reproducibility of $\text{\upshape Sk}_{\mu_1,\mu_2}$ \cite{29}]\label{sksum} Let $X\leftarrow \text{\upshape Sk}_{\mu_1,\mu_2}$ and $Y\leftarrow \text{\upshape Sk}_{\mu_3,\mu_4}$ be independent random variables. Then $Z:=X+Y$ is distributed according to $\text{\upshape Sk}_{\mu_1+\mu_3,\mu_2+\mu_4}$.
\end{Lem}

An induction step shows that the sum of $n$ i.i.d. symmetric Skellam random variables with variance $\mu$ is a symmetric Skellam random variable with variance $n\mu$. For our analysis, we will use the following bound on the ratio of modified Bessel functions of the first kind. 

\begin{Lem}[Bound on $I_k(\mu)$ \cite{27}]\label{modbesrat} For real $k$, let $I_k(\mu)$ be the modified Bessel function of the first kind and order $k$. Then
 \[1>\frac{I_k(\mu)}{I_{k-1}(\mu)}>\frac{-k+\sqrt{k^2+\mu^2}}{\mu},\, k\geq 0.\]
\end{Lem}

Moreover, we need a proper bound on the symmetric Skellam distribution that holds with probability exponentially close to $1$.

\begin{Lem}[Bound on the Skellam distribution]\label{skellbound}
Let $\kappa$ be a security parameter, let $s=s(\kappa)=\omega(\log(\kappa))$ and let $\mu=\mu(\kappa)=\text{\upshape\sffamily poly}(\kappa)$ with $\mu>s>0$. Let $X\leftarrow\text{\upshape Sk}_\mu$. Then \[\Pr[X>s\sqrt{\mu}]\leq\text{\upshape\sffamily neg}(\kappa).\]
\end{Lem}
\begin{proof}
Applying the Laplace transform and the Markov's inequality we obtain for any $t>0$,
\[\Pr[X>s\sqrt{\mu}]=\Pr[e^{tX}>e^{ts\sqrt{\mu}}]\leq\frac{\E[e^{tX}]}{e^{ts\sqrt{\mu}}}.\]
As shown in \cite{30}, the moment generating function of $X\sim \text{Sk}(\mu)$ is 
\[\E[e^{tX}]=e^{-\mu(1-\cosh(t))},\]
where $\cosh(t)= (e^t + e^{-t})/2$. Hence, we have
\[\Pr[X>s\sqrt{\mu}]\leq e^{-\mu(1-\cosh(t))-ts\sqrt{\mu}}< e^{-s(1-\cosh(t))-ts^{3/2}}.\]
Fix $t=\operatorname{arsinh}(1/\sqrt{s})$. Then
\begin{align*} \Pr[X>s\sqrt{\mu}]< & e^{-s(1-\sqrt{1+1/s})-s^{3/2}\operatorname{arsinh}(1/\sqrt{s})}\\
= & e^{-s}\cdot e^{s\cdot(\sqrt{1+1/s}-\sqrt{s}\operatorname{arsinh}(1/\sqrt{s}))}\\
< & e^{-s}\cdot e^{2/3}\\
= & \text{\upshape\sffamily neg}(\kappa).
\end{align*}
To see the last inequality, observe that the function $f(s)=s\cdot(\sqrt{1+1/s}-\sqrt{s}\operatorname{arsinh}(1/\sqrt{s}))$ is monotonically increasing and its limit is $2/3$.
\end{proof}

\subsection{Facts about Learning with Errors}

In \cite{42}, Regev established worst-to-average-case connections between conjecturally hard lattice problems and the $\text{\upshape LWE}(\kappa,\lambda,q,D_\nu)$ problem.

\begin{Thm}[Worst-to-Average Case \cite{42}]\label{wtacr} Let $\kappa$ be a security parameter and let $q = q(\kappa)$ be a modulus, let $\alpha=\alpha(\kappa)\in(0,1)$ be such that $\alpha q> 2\sqrt{\kappa}$. If there exists a probabilistic polynomial-time algorithm solving the $\text{\upshape LWE}(\kappa,\lambda,q,D_{(\alpha q)^2/(2\pi)})$ problem with non-negligible probability, 
then there exists an efficient quantum algorithm that approximates the decision-version of the shortest vector problem ($\text{\upshape GAPSVP}$) and the shortest independent vectors problem ($\text{\upshape SIVP}$) to within $\tilde{O}(\kappa/\alpha)$ in the worst case.
\end{Thm}

We use the search-to-decision reduction from \cite{43} basing the hardness of Problem $2$ on the hardness of Problem $1$ which works for any error distribution $\chi$ and is sample preserving.

\begin{Thm}[Search-to-Decision \cite{43}]\label{lwestod} Let $\kappa$ be a security parameter, $q = q(\kappa) = \text{\upshape\sffamily poly}(\kappa)$ a prime modulus and let $\chi$ be any distribution on $\mathbb{Z}_q$. Assume there exists a probabilistic polynomial-time distinguisher that solves the $\text{\upshape DLWE}(\kappa,\lambda,q,\chi)$ problem with non-negligible success-probability, then there exists a probabilistic polynomial-time adversary that solves the  $\text{\upshape LWE}(\kappa,\lambda,q,\chi)$ problem with non-negligible success-probability.
\end{Thm}

Finally, we provide a matrix version of Problem $2$. The hardness of this version can be shown by using a hybrid argument as pointed out in \cite{39}.

\begin{Lem}[Matrix version of LWE]\label{lwematver}
Let $\kappa$ be a security parameter, $\lambda=\lambda(\kappa)= \text{\upshape\sffamily poly}(\kappa)$, $\kappa^\prime=\kappa^\prime(\kappa)=\text{\upshape\sffamily poly}(\kappa)$. Assume that the $\text{\upshape DLWE}(\kappa,\lambda,q,\chi)$ problem is hard. Then $(\textbf{\upshape A},\textbf{\upshape AX}+\textbf{\upshape E})$ is pseudo-random, where $\textbf{\upshape A}\leftarrow\mathcal{U}(\mathbb{Z}_q^{\lambda\times\kappa}), \textbf{\upshape X}\leftarrow\mathcal{U}(\mathbb{Z}_q^{\kappa\times\kappa^\prime})$ and $\textbf{\upshape E}\leftarrow\chi^{\lambda\times\kappa^\prime}$.
\end{Lem}

\subsection{Facts about Lossy Codes}

We will use the fact that the existence of a lossy code for an error distribution implies the hardness of the associated decoding problems. This means, solving the LWE problem is hard, even though with overwhelming probability the secret is information-theoretically unique. In \cite{39} it was shown that solving the LWE problem for $\chi$ is hard if there exists a lossy code for $\chi$ in the sense of \cite{39}. Here we prove this statement for our definition of lossy codes. The proof is very similar to the one in \cite{39}.

\begin{Thm}[Lossy code gives hard LWE]\label{lossythm}
Let $\kappa$ be a security parameter, let $\lambda=\lambda(\kappa)=\text{\upshape\sffamily poly}(\kappa)$ and let $q=q(\kappa)$ be a modulus. Let the distribution $\chi$ on $\mathbb{Z}_q$ be efficiently samplable. Let $\Delta=\Delta(\kappa)=\omega(\log(\kappa))$. Then the $\text{\upshape LWE}(\kappa,\lambda,q,\chi)$ problem is hard, given that there exists a family $\{\mathcal{C}_{\kappa,\lambda,q}\}\subseteq\mathbb{Z}_q^{\lambda\times\kappa}$ of $\Delta$-lossy codes for the error distribution $\chi$.
\end{Thm}
\begin{proof}
Due to the non-lossiness of $\mathcal{U}(\mathbb{Z}_q^{\lambda\times\kappa})$ for $\chi$, instances of $\text{\upshape LWE}(\kappa,\lambda,q,\chi)$ have a unique solution with overwhelming probability. Now, let $\mathcal{T}$ be a probabilistic polynomial-time adversary solving the $\text{\upshape LWE}(\kappa,\lambda,q,\chi)$ problem with non-negligible probability $\sigma$. Using $\mathcal{T}$, we will construct a probabilistic polynomial-time distinguisher $\mathcal{D}_{\mbox{\scriptsize LWE}}$ distinguishing $\mathcal{U}(\mathbb{Z}_q^{\lambda\times\kappa})$ and $\{\mathcal{C}_{\kappa,\lambda,q}\}$ with non-negligible advantage.\\
Let $\mathcal{D}_{\mbox{\scriptsize LWE}}$ be given $\textbf{\upshape A}\in\mathbb{Z}^{\lambda\times\kappa}$ as input. It must decide, whether $\textbf{\upshape A}\leftarrow\mathcal{U}(\mathbb{Z}_q^{\lambda\times\kappa})$ or $\textbf{\upshape A}\leftarrow\{\mathcal{C}_{\kappa,\lambda,q}\}$. Therefore, $\mathcal{D}_{\mbox{\scriptsize LWE}}$ samples $\tilde{\textbf{\upshape x}}\leftarrow\mathcal{U}(\mathbb{Z}^\kappa)$ and $\tilde{\textbf{\upshape e}}\leftarrow\chi^\lambda$. It runs $\mathcal{T}$ on input $(\textbf{\upshape A},\textbf{\upshape A}\tilde{\textbf{\upshape x}}+\tilde{\textbf{\upshape e}})$. Then $\mathcal{T}$ outputs some $\textbf{\upshape x}\in\mathbb{Z}^\kappa$. If $\textbf{\upshape x}=\tilde{\textbf{\upshape x}}$, then $\mathcal{D}_{\mbox{\scriptsize LWE}}$ outputs $1$, otherwise it outputs $0$.\\
If $\textbf{\upshape A}\leftarrow\mathcal{U}(\mathbb{Z}_q^{\lambda\times\kappa})$, then $\tilde{\textbf{\upshape x}}$ is unique and then $\textbf{\upshape x}=\tilde{\textbf{\upshape x}}$ with probability $\sigma$. Therefore 
\[\Pr[\mathcal{D}_{\mbox{\scriptsize LWE}}(\textbf{\upshape A})=1\,|\,\textbf{\upshape A}\leftarrow\mathcal{U}(\mathbb{Z}_q^{\lambda\times\kappa})]=\sigma.\] 
If $\textbf{\upshape A}\leftarrow\{\mathcal{C}_{\kappa,\lambda,q}\}$, then $\mathcal{T}$ outputs the correct value with probability \[\Pr[\textbf{x}=\tilde{\textbf{x}}\,|\,\textbf{Ax} + \textbf{e}=\textbf{A}\tilde{\textbf{x}} + \tilde{\textbf{e}}]=2^{-H_{f_{\textbf{\upshape A},\textbf{\upshape e}},f_{\textbf{\upshape A},\tilde{\textbf{\upshape e}}},\tilde{\textbf{\upshape x}}}(\textbf{\upshape x})}\leq 2^{-\Delta},\] 
with $f_{\textbf{\upshape B},\textbf{\upshape b}}(\textbf{\upshape y})=\textbf{\upshape B}\textbf{\upshape y} + \textbf{\upshape b}$. This holds with overwhelming probability over the choice of $(\textbf{\upshape A},\tilde{\textbf{\upshape x}},\tilde{\textbf{\upshape e}})$. This probability is negligible in $\kappa$, since $\Delta=\omega(\log(\kappa))$. Therefore \[\Pr[\mathcal{D}_{\mbox{\scriptsize LWE}}(\textbf{\upshape A})=1\,|\,\textbf{\upshape A}\leftarrow\{\mathcal{C}_{\kappa,\lambda,q}\}]=\text{\upshape\sffamily neg}(\kappa)\] 
and in conclusion $\mathcal{D}_{\mbox{\scriptsize LWE}}$ distinguishes $\mathcal{U}(\mathbb{Z}_q^{\lambda\times\kappa})$ and $\{\mathcal{C}_{\kappa,\lambda,q}\}$ with probability at least $\sigma-\text{\upshape\sffamily neg}(\kappa)$, which is non-negligible.
\end{proof}


We will use the fact that $\mathcal{U}(\mathbb{Z}_q^{\lambda\times\kappa})$ is always non-lossy if the corresponding error distribution $\chi$ can be bounded.

\begin{Lem}[Non-lossiness of $\mathcal{U}(\mathbb{Z}_q^{\lambda\times\kappa})$]\label{nonloss}
Let $\kappa$ be a security parameter and $\chi$ a probability distribution on $\mathbb{Z}$. Say, the support of $\chi$ can be bounded by $r=r(\kappa)=\text{\upshape\sffamily poly}(\kappa)$. Moreover, let $q>(4r+1)^{1+\tau}$ for a constant $\tau>0$ and $\lambda=\lambda(\kappa)>(1+2/\tau)\kappa$. Let $f_{\textbf{\upshape B},\textbf{\upshape b}}(\textbf{\upshape y})=\textbf{\upshape B}\textbf{\upshape y} + \textbf{\upshape b}$. Let $\textbf{\upshape A}\leftarrow\mathcal{U}(\mathbb{Z}_q^{\lambda\times\kappa})$, $\tilde{\textbf{\upshape x}}\leftarrow\mathcal{U}(\mathbb{Z}_q^\kappa), \tilde{\textbf{\upshape e}}\leftarrow\chi^\lambda$, let $\textbf{\upshape x}\leftarrow\mathcal{U}(\mathbb{Z}_q^\kappa)$ and $\textbf{\upshape e}\leftarrow\chi^\lambda$. Then
\[\Pr_{(\textbf{\upshape A},\tilde{\textbf{\upshape x}},\tilde{\textbf{\upshape e}})}[H_{f_{\textbf{\upshape A},\textbf{\upshape e}},f_{\textbf{\upshape A},\tilde{\textbf{\upshape e}}},\tilde{\textbf{\upshape x}}}(\textbf{\upshape x})=0]\geq 1-\text{\upshape\sffamily neg}(\kappa).\]
\end{Lem}

The proof of Lemma \ref{nonloss} is identical to the proof of Lemma $3$ in \cite{39}: since $\Pr_{\textbf{\upshape A}}[||\textbf{\upshape A}\tilde{\textbf{\upshape x}}||_\infty\leq 2r]\leq\text{\upshape\sffamily neg}(\kappa)$ (as shown in \cite{39}) and $||\tilde{\textbf{\upshape e}}-\textbf{\upshape e}||_\infty\leq 2r$ for independent $\textbf{\upshape e},\tilde{\textbf{\upshape e}}\leftarrow\chi^\lambda$ (where the norm is computed in the central residue-class representation of the elements in $\mathbb{Z}_q$), with probability $1-\text{\upshape\sffamily neg}(\kappa)$ there cannot exist any $\textbf{\upshape x}\in\mathbb{Z}_q^\kappa$ with $\textbf{\upshape x}\neq\tilde{\textbf{\upshape x}}$ and $\textbf{\upshape A}\textbf{\upshape x}+\textbf{\upshape e}=\textbf{\upshape A}\tilde{\textbf{\upshape x}}+\tilde{\textbf{\upshape e}}$.

\section{Proof of the Hardness Result}\label{hardn}

Now we construct the lossy code for the Skellam distribution. It is essentially the same construction that was used as lossy code for the uniform error distribution in \cite{39}.

\begin{Ctn}[Lossy code for the symmetric Skellam distribution]\label{lossysk}
Let $\kappa$ be an even security parameter, let $\lambda=\lambda(\kappa)=\text{\upshape\sffamily poly}(\kappa)$, $\nu>0$ and let $q=q(\kappa)$ be a prime modulus. The distribution $\mathcal{C}_{\kappa,\lambda,q,\nu}$ defined on $\mathbb{Z}_q^{\lambda\times\kappa}$ is specified as follows. 
Choose $\textbf{\upshape A}^\prime\leftarrow\mathcal{U}(\mathbb{Z}_q^{\lambda\times\kappa/2})$, $\textbf{\upshape T}\leftarrow\mathcal{U}(\mathbb{Z}_q^{\kappa/2\times\kappa/2})$ and $\textbf{\upshape G}\leftarrow D_\nu^{\lambda\times\kappa/2}$. 
Output \[\textbf{\upshape A}=(\textbf{\upshape A}^\prime||\textbf{\upshape A}^\prime\textbf{\upshape T}+\textbf{\upshape G}).\] 
\end{Ctn}

From Lemma \ref{lwematver} and Theorem \ref{lwestod} it is straightforward that $\mathcal{C}_{\kappa,\lambda,q,\nu}$ is pseudo-random in the sense of property $1$ of Definition \ref{lossydef} assuming the hardness of the $\text{\upshape LWE}(\kappa,\lambda,q,D_\nu)$ problem.

\begin{Lem}[Pseudo-randomness of Construction $\ref{lossysk}$ \cite{39}]\label{lossypr}
For Construction $\ref{lossysk}$ it holds that $\mathcal{C}_{\kappa,\lambda,q,\nu}\approx_c\mathcal{U}(\mathbb{Z}_q^{\lambda\times\kappa})$ assuming the hardness of the $\text{\upshape LWE}(\kappa,\lambda,q,D_\nu)$ problem.
\end{Lem}

Let $\textbf{\upshape A}=(\textbf{\upshape A}^\prime||\textbf{\upshape A}^\prime\textbf{\upshape T}+\textbf{\upshape G})$ be the code as defined in Construction \ref{lossysk}. We show that in our further analysis we can restrict ourselves to considering only $\textbf{\upshape G}$ instead of $\textbf{\upshape A}$.

\begin{Lem}\label{ginsteadofa}
Let $\kappa$ be an even integer. Let $\textbf{\upshape A}=(\textbf{\upshape A}^\prime||\textbf{\upshape A}^\prime\textbf{\upshape T}+\textbf{\upshape G})$ with $\textbf{\upshape A}^\prime\in\mathbb{Z}_q^{\lambda\times\kappa/2}$, $\textbf{\upshape T}\in\mathbb{Z}_q^{\kappa/2\times\kappa/2}$ and $\textbf{\upshape G}\in\mathbb{Z}_q^{\lambda\times\kappa/2}$. Then for all $\textbf{\upshape x}\in\mathbb{Z}_q^{\kappa/2}$ there is a $\textbf{\upshape x}^\prime\in\mathbb{Z}_q^\kappa$ with $\textbf{\upshape Ax}^\prime=\textbf{\upshape Gx}$.
\end{Lem}
\begin{proof}
Define \[\widetilde{\textbf{\upshape T}}=\begin{pmatrix} \textbf{\upshape I} & \textbf{\upshape T}\\ \textbf{\upshape 0} & \textbf{\upshape I}\end{pmatrix}.\] Note that $\widetilde{\textbf{\upshape T}}\in\mathbb{Z}_q^{\kappa\times\kappa}$ is a regular matrix. For all $\textbf{\upshape x}\in\mathbb{Z}_q^{\kappa/2}$ set $\textbf{\upshape x}^\prime=\widetilde{\textbf{\upshape T}}^{-1}\cdot(\textbf{0}^{\text{\scriptsize tr}}||\textbf{\upshape x}^{\text{\scriptsize tr}})^{\text{\scriptsize tr}}$. Then 
\begin{align*} \textbf{\upshape Ax}^\prime = & (\textbf{\upshape A}^\prime||\textbf{\upshape A}^\prime\textbf{\upshape T}+\textbf{\upshape G})\cdot\textbf{\upshape x}^\prime\\
= & (\textbf{\upshape A}^\prime||\textbf{\upshape G})\cdot\widetilde{\textbf{\upshape T}}\cdot\widetilde{\textbf{\upshape T}}^{-1}\cdot(\textbf{0}^{\text{\scriptsize tr}}||\textbf{\upshape x}^{\text{\scriptsize tr}})^{\text{\scriptsize tr}}\\
= & (\textbf{\upshape A}^\prime||\textbf{\upshape G})\cdot(\textbf{0}^{\text{\scriptsize tr}}||\textbf{\upshape x}^{\text{\scriptsize tr}})^{\text{\scriptsize tr}}\\
= & \textbf{\upshape Gx}.
\end{align*}
\end{proof}

To show that Construction \ref{lossysk} is a lossy code for the symmetric Skellam distribution we prove that the second property of Definition \ref{lossydef} is satisfied. We first prove two supporting claims and then show the lossiness of Construction \ref{lossysk}.

\begin{Lem}\label{techlem}
$-C+\sqrt{C^2+1}\geq\exp(-C)$ for all $C\geq 0$.
\end{Lem}
\begin{proof}
Let $f(C)=(-C+\sqrt{C^2+1})\exp(C)$. Then $f(C)$ is monotonically increasing and $f(0)=1\cdot(-0+\sqrt{0+1})=1$.
\end{proof}

\begin{Lem}\label{smallnormvec}
Let $\kappa$ be a security parameter, let $s=s(\kappa)=\omega(\log(\kappa))$ and let $\nu=\nu(\kappa)=\text{\upshape\sffamily poly}(\kappa)$. Let $\lambda=\lambda(\kappa), \zeta=\zeta(\kappa)$ be integers. Let $\textbf{\upshape G}\leftarrow D_\nu^{\lambda\times\zeta}$. Then for all $\textbf{\upshape z}\in\{0,1\}^\zeta$ the following hold:
\begin{enumerate}
\item $\Pr[||\textbf{\upshape Gz}||_1>\lambda\sqrt{\zeta\nu s}]\leq\text{\upshape\sffamily neg}(\kappa).$
\item $\Pr[||\textbf{\upshape Gz}||_2^2>\lambda\zeta\nu s]\leq\text{\upshape\sffamily neg}(\kappa).$
\end{enumerate}
\end{Lem}
\begin{proof}
The claims follow from Lemma \ref{disgausbound}.
\end{proof}

\begin{Lem}[Lossiness of Construction $\ref{lossysk}$]\label{lossysklem}
Let $\kappa$ be an even security parameter, $s=s(\kappa)=\omega(\log(\kappa))$ with $\kappa>s^3$, let $\nu=\nu(\kappa)$, let $q=\text{\upshape\sffamily poly}(\kappa)$ be a sufficiently large prime modulus, let $\lambda=\lambda(\kappa)> s$ and let $\Delta=\Delta(\kappa)=\omega(\log(\kappa))$. Let $\mu=\mu(\kappa)\geq\lambda^2\nu$. Let $f_{\textbf{\upshape B},\textbf{\upshape b}}(\textbf{\upshape y})=\textbf{\upshape B}\textbf{\upshape y} + \textbf{\upshape b}$. Let $\textbf{\upshape A}\leftarrow\{\mathcal{C}_{\kappa,\lambda,q,\nu}\}$, let $\tilde{\textbf{\upshape x}}\leftarrow\mathcal{U}(\mathbb{Z}_q^\kappa), \tilde{\textbf{\upshape e}}\leftarrow\text{\upshape Sk}_\mu^\lambda$, let $\textbf{\upshape x}\leftarrow\mathcal{U}(\mathbb{Z}_q^\kappa)$ and $\textbf{\upshape e}\leftarrow\text{\upshape Sk}_\mu^\lambda$. Then
\[\Pr_{(\textbf{\upshape A},\tilde{\textbf{\upshape x}},\tilde{\textbf{\upshape e}})}[H_{f_{\textbf{\upshape A},\textbf{\upshape e}},f_{\textbf{\upshape A},\tilde{\textbf{\upshape e}}},\tilde{\textbf{\upshape x}}}(\textbf{\upshape x})\geq\Delta]\geq 1-\text{\upshape\sffamily neg}(\kappa).\]
\end{Lem}
\begin{proof}
Let $(\textbf{\upshape Mz})_j$ denote the $j$-th entry of $\textbf{\upshape Mz}$ for a matrix $\textbf{\upshape M}$ and a vector $\textbf{\upshape z}$. Let $\textbf{\upshape A}=(\textbf{\upshape A}^\prime||\textbf{\upshape A}^\prime\textbf{\upshape T}+\textbf{\upshape G})$ be distributed according to $\{\mathcal{C}_{\kappa,\lambda,q,\nu}\}$ with $\textbf{\upshape A}^\prime\leftarrow\mathcal{U}(\mathbb{Z}_q^{\lambda\times\kappa/2})$, $\textbf{\upshape T}\leftarrow\mathcal{U}(\mathbb{Z}_q^{\kappa/2\times\kappa/2})$ and $\textbf{\upshape G}\leftarrow D_\nu^{\lambda\times\kappa/2}$. Let $\tilde{\textbf{\upshape e}}=(\tilde{e}_j)_{j=1,\ldots,\lambda}\leftarrow\text{\upshape Sk}_{\mu}^\lambda$. Then we have the following chain of (in)equations:
\begin{align*}
 & \Pr_{(\textbf{\upshape A},\tilde{\textbf{\upshape x}},\tilde{\textbf{\upshape e}})}[H_{f_{\textbf{\upshape A},\textbf{\upshape e}},f_{\textbf{\upshape A},\tilde{\textbf{\upshape e}}},\tilde{\textbf{\upshape x}}}(\textbf{\upshape x})\geq\Delta]\\
= & \Pr_{(\textbf{\upshape A},\tilde{\textbf{\upshape x}},\tilde{\textbf{\upshape e}})}\left[\Pr[\textbf{x}=\tilde{\textbf{x}}\,|\,\textbf{Ax} + \textbf{e}=\textbf{A}\tilde{\textbf{x}} + \tilde{\textbf{e}}]\leq 2^{-\Delta}\right]\\
= & \Pr_{(\textbf{\upshape A},\tilde{\textbf{\upshape x}},\tilde{\textbf{\upshape e}})}\left[\Pr_{\textbf{\upshape e}}[\textbf{Ax} + \textbf{e}=\textbf{A}\tilde{\textbf{x}} + \tilde{\textbf{e}}\,|\,\textbf{\upshape x}=\tilde{\textbf{x}}]\cdot\frac{\Pr[\textbf{x}=\tilde{\textbf{x}}]}{\Pr_{(\textbf{x},\textbf{e})}[\textbf{\upshape A}\textbf{\upshape x}+\textbf{\upshape e}=\textbf{\upshape A}\tilde{\textbf{x}}+\tilde{\textbf{e}}]}\leq 2^{-\Delta}\right]\numberthis\label{bayeseq}\\
= & \Pr_{(\textbf{\upshape A},\tilde{\textbf{\upshape x}},\tilde{\textbf{\upshape e}})}\left[\Pr_{\textbf{\upshape e}}[\textbf{Ax} + \textbf{e}=\textbf{A}\tilde{\textbf{x}} + \tilde{\textbf{e}}\,|\,\textbf{\upshape x}=\tilde{\textbf{x}}]\cdot\right.\\ 
& \mbox{\,\,\,\,\,\,\,\,\,\,\,\,\,\,\,\,\,\,\,\,\,\,\,\,\,\,\,\,}\left.\cdot\frac{\Pr[\textbf{\upshape x}=\tilde{\textbf{x}}]}{\sum_{\textbf{\upshape z}\in\mathbb{Z}_q^\kappa}\Pr_{\textbf{\upshape e}}[\textbf{\upshape A}\textbf{\upshape x}+\textbf{\upshape e}=\textbf{\upshape A}\tilde{\textbf{\upshape x}}+\tilde{\textbf{\upshape e}}\,|\,\textbf{\upshape x}=\textbf{\upshape z}]\cdot\Pr[\textbf{\upshape x}=\textbf{\upshape z}]}\leq 2^{-\Delta}\right]\\
= & \Pr_{(\textbf{\upshape A},\tilde{\textbf{\upshape x}},\tilde{\textbf{\upshape e}})}\left[\Pr_{\textbf{\upshape e}}[\textbf{\upshape e}=\tilde{\textbf{\upshape e}}]\cdot\frac{\Pr[\textbf{\upshape x}=\tilde{\textbf{\upshape x}}]}{\sum_{\textbf{\upshape z}\in\mathbb{Z}_q^\kappa}\Pr_{\textbf{\upshape e}}[\textbf{\upshape A}(\textbf{\upshape z}-\tilde{\textbf{\upshape x}})+\textbf{\upshape e}=\tilde{\textbf{\upshape e}}]\cdot\Pr[\textbf{\upshape x}=\textbf{\upshape z}]}\leq 2^{-\Delta}\right]\\
= & \Pr_{(\textbf{\upshape A},\tilde{\textbf{\upshape x}},\tilde{\textbf{\upshape e}})}\left[\Pr_{\textbf{\upshape e}}[\textbf{\upshape e}=\tilde{\textbf{\upshape e}}]\cdot\frac{1}{\sum_{\textbf{\upshape z}\in\mathbb{Z}_q^\kappa}\Pr_{\textbf{\upshape e}}[\textbf{\upshape A}(\textbf{\upshape z}-\tilde{\textbf{\upshape x}})+\textbf{\upshape e}=\tilde{\textbf{\upshape e}}]}\leq 2^{-\Delta}\right]\numberthis\label{applyuniform}\\
= & \Pr_{(\textbf{\upshape A},\tilde{\textbf{\upshape e}})}\left[\frac{\Pr_{\textbf{\upshape e}}[\textbf{\upshape e}=\tilde{\textbf{\upshape e}}]}{\sum_{\textbf{\upshape z}\in\mathbb{Z}_q^\kappa}\Pr_{\textbf{\upshape e}}[\textbf{\upshape e}=\textbf{\upshape A}\textbf{\upshape z}+\tilde{\textbf{\upshape e}}]}\leq 2^{-\Delta}\right]\numberthis\label{eliminatevariable}\\
= & \Pr_{(\textbf{\upshape A},\tilde{\textbf{\upshape e}})}\left[\sum_{\textbf{\upshape z}\in\mathbb{Z}_q^\kappa}\frac{\Pr_{\textbf{\upshape e}}[\textbf{\upshape e}=\textbf{\upshape A}\textbf{\upshape z}+\tilde{\textbf{\upshape e}}]}{\Pr_{\textbf{\upshape e}}[\textbf{\upshape e}=\tilde{\textbf{\upshape e}}]}\geq 2^{\Delta}\right]\\
= & \Pr_{(\textbf{\upshape A},\tilde{\textbf{\upshape e}})}\left[\sum_{\textbf{\upshape z}\in\mathbb{Z}_q^\kappa}\frac{\prod_{j=1}^\lambda \exp(-\mu)\cdot I_{(\textbf{\upshape A}\textbf{\upshape z})_j+\tilde{e}_j}(\mu)}{\prod_{j=1}^\lambda \exp(-\mu)\cdot I_{\tilde{e}_j}(\mu)}\geq 2^{\Delta}\right]\\
= & \Pr_{(\textbf{\upshape A},\tilde{\textbf{\upshape e}})}\left[\sum_{\textbf{\upshape z}\in\mathbb{Z}_q^\kappa}\prod_{j=1}^\lambda\frac{I_{(\textbf{\upshape A}\textbf{\upshape z})_j+\tilde{e}_j}(\mu)}{I_{\tilde{e}_j}(\mu)}\geq 2^{\Delta}\right]\\
\geq & \Pr_{(\textbf{\upshape A},\tilde{\textbf{\upshape e}})}\left[\sum_{\textbf{\upshape z}\in\mathbb{Z}_q^\kappa}\prod_{j=1}^\lambda\prod_{k=1+\tilde{e}_j}^{(\textbf{\upshape A}\textbf{\upshape z})_j+\tilde{e}_j}\frac{-k+\sqrt{k^2+\mu^2}}{\mu}\geq 2^{\Delta}\right]\numberthis\label{approxmodbes}\\
= & \Pr_{(\textbf{\upshape A},\tilde{\textbf{\upshape e}})}\left[\sum_{\textbf{\upshape z}\in\mathbb{Z}_q^\kappa}\prod_{j=1}^\lambda\prod_{k=1}^{(\textbf{\upshape A}\textbf{\upshape z})_j}\left(\frac{-(k+\tilde{e}_j)}{\mu}+\sqrt{\left(\frac{k+\tilde{e}_j}{\mu}\right)^2+1}\right)\geq 2^{\Delta}\right]\\
\geq & \Pr_{(\textbf{\upshape A},\tilde{\textbf{\upshape e}})}\left[\sum_{\textbf{\upshape z}\in\mathbb{Z}_q^\kappa}\prod_{j=1}^\lambda\left(\frac{-((\textbf{\upshape A}\textbf{\upshape z})_j+\tilde{e}_j)}{\mu}+\sqrt{\left(\frac{(\textbf{\upshape A}\textbf{\upshape z})_j+\tilde{e}_j}{\mu}\right)^2+1}\right)^{(\textbf{\upshape A}\textbf{\upshape z})_j}\geq 2^{\Delta}\right]\numberthis\label{mondeceq}\\
\geq & \Pr_{(\textbf{\upshape A},\tilde{\textbf{\upshape e}})}\left[\sum_{\textbf{\upshape z}\in\mathbb{Z}_q^\kappa}\prod_{j=1}^\lambda\exp\left(-\frac{(\textbf{\upshape A}\textbf{\upshape z})_j+\tilde{e}_j}{\mu}\right)^{(\textbf{\upshape A}\textbf{\upshape z})_j}\geq 2^{\Delta}\right]\numberthis\label{techlemeq}\\
= & \Pr_{(\textbf{\upshape A},\tilde{\textbf{\upshape e}})}\left[\sum_{\textbf{\upshape z}\in\mathbb{Z}_q^\kappa}\prod_{j=1}^\lambda\exp\left(-\frac{(\textbf{\upshape A}\textbf{\upshape z})_j^2+(\textbf{\upshape A}\textbf{\upshape z})_j\cdot\tilde{e}_j}{\mu}\right)\geq 2^{\Delta}\right]\\
\geq & \Pr_{\textbf{\upshape A}}\left[\sum_{\textbf{\upshape z}\in\mathbb{Z}_q^\kappa}\prod_{j=1}^\lambda\exp\left(-\frac{(\textbf{\upshape A}\textbf{\upshape z})_j^2+(\textbf{\upshape A}\textbf{\upshape z})_j\cdot s\sqrt{\mu}}{\mu}\right)\geq 2^{\Delta}\right]-\text{\upshape\sffamily neg}(\kappa)\numberthis\label{skellamboundeq}\\ 
= & \Pr_{\textbf{\upshape A}}\left[\sum_{\textbf{\upshape z}\in\mathbb{Z}_q^\kappa}\exp\left(-\frac{\sum_{j=1}^\lambda(\textbf{\upshape A}\textbf{\upshape z})_j^2+s\sqrt{\mu}\cdot\sum_{j=1}^\lambda(\textbf{\upshape A}\textbf{\upshape z})_j}{\mu}\right)\geq 2^{\Delta}\right]-\text{\upshape\sffamily neg}(\kappa)\\
= & \Pr_{\textbf{\upshape A}}\left[\sum_{\textbf{\upshape z}\in\mathbb{Z}_q^\kappa}\exp\left(-\frac{||\textbf{\upshape A}\textbf{\upshape z}||_2^2+s\sqrt{\mu}\cdot||\textbf{\upshape A}\textbf{\upshape z}||_1}{\mu}\right)\geq 2^{\Delta}\right]-\text{\upshape\sffamily neg}(\kappa)\\
\geq & \Pr_{\textbf{\upshape G}}\left[\sum_{\textbf{\upshape z}\in\mathbb{Z}_q^{\kappa/2}}\exp\left(-\frac{||\textbf{\upshape Gz}||_2^2+s\sqrt{\mu}\cdot||\textbf{\upshape Gz}||_1}{\mu}\right)\geq 2^{\Delta}\right]-\text{\upshape\sffamily neg}(\kappa).\numberthis\label{fromAtoG}
\end{align*}

Equation \ref{bayeseq} is an application of the Bayes rule and Equation \ref{applyuniform} applies, since $\textbf{\upshape x}$ is sampled according to a uniform distribution. Equation \ref{eliminatevariable} is valid since in the denominator we are summing over all possible $\textbf{\upshape z}\in\mathbb{Z}_q^\kappa$. Inequation \ref{approxmodbes} is an iterative application of Theorem \ref{modbesrat}. Note that the modified Bessel function of the first kind ist symmetric when considered over integer orders. Therefore, from this point of the chain of (in)equations (i.e. from Inequation \ref{approxmodbes}), we can assume that $\tilde{e}_j\geq 0$. Moreover, we can assume that $(\textbf{\upshape A}\textbf{\upshape z})_j\geq 0$, since otherwise $I_{(\textbf{\upshape A}\textbf{\upshape z})_j+\tilde{e}_j}(\mu)>I_{-(\textbf{\upshape A}\textbf{\upshape z})_j+\tilde{e}_j}(\mu)$. I.e. if $(\textbf{\upshape A}\textbf{\upshape z})_j< 0$, then we implicitly change the sign of the $j$th row in the original matrix $\textbf{\upshape A}$ while considering the particular $\textbf{\upshape z}$. In this way we are always considering the worst-case scenario for every $\textbf{\upshape z}$. Note that this step does not change the distribution of $\textbf{\upshape A}$, since $\{\mathcal{C}_{\kappa,\lambda,q,\nu}\}$ is symmetric. Inequation \ref{mondeceq} holds, since $f_\mu(k)=(-k+\sqrt{k^2+\mu^2})/\mu$ is a monotonically decreasing function. Inequation \ref{techlemeq} follows from Lemma \ref{techlem} by setting $C=((\textbf{\upshape A}\textbf{\upshape z})_j+\tilde{e}_j)/\mu$.  
Inequation \ref{skellamboundeq} holds because of the bound in Lemma \ref{skellbound}. Inequation \ref{fromAtoG} follows from Lemma \ref{ginsteadofa}, since $\textbf{\upshape A}=(\textbf{\upshape A}^\prime||\textbf{\upshape A}^\prime\textbf{\upshape T}+\textbf{\upshape G})$.\\
Now consider the set $\mathcal{Z}\subset\{0,1\}^{\kappa/2}$ with each element in $\mathcal{Z}$ having hamming weight exactly $\kappa/4$. Then $|\mathcal{Z}|=\binom{\kappa/2}{\kappa/4}>2^{\kappa/4}$. Since $\mu=\lambda^2\nu$, from Lemma \ref{smallnormvec} it follows that
\[\Pr_{\textbf{\upshape G}}\left[\sum_{\textbf{\upshape z}\in\mathcal{Z}}\exp\left(-\frac{||\textbf{\upshape Gz}||_2^2+s\sqrt{\mu}\cdot||\textbf{\upshape Gz}||_1}{\mu}\right)\geq 2^{\kappa/4}\cdot\exp\left(-\frac{\kappa s}{4\lambda}-\frac{s\sqrt{\kappa s}}{2}\right)\right]\geq 1-\text{\upshape\sffamily neg}(\kappa),\]
where the norm is computed in the central residue-class representation of the elements in $\mathbb{Z}_q$. Moreover we have
\[2^{\kappa/4}\cdot\exp\left(-\frac{\kappa s}{4\lambda}-\frac{s\sqrt{\kappa s}}{2}\right)>C^\kappa\]
for some constant $C>1$, since $\kappa> s^3$ and $\lambda> s$. Therefore
\begin{align*}
 & \Pr_{(\textbf{\upshape A},\tilde{\textbf{\upshape x}},\tilde{\textbf{\upshape e}})}[H_{f_{\textbf{\upshape A},\textbf{\upshape e}},f_{\textbf{\upshape A},\tilde{\textbf{\upshape e}}},\tilde{\textbf{\upshape x}}}(\textbf{\upshape x})\geq\Delta]\\ 
 \geq & \Pr_{\textbf{\upshape G}}\left[\sum_{\textbf{\upshape z}\in\mathbb{Z}_q^{\kappa/2}}\exp\left(-\frac{||\textbf{\upshape Gz}||_2^2+s\sqrt{\mu}\cdot||\textbf{\upshape Gz}||_1}{\mu}\right)\geq 2^\Delta\right]-\text{\upshape\sffamily neg}(\kappa)\\
\geq & \Pr_{\textbf{\upshape G}}\left[\sum_{\textbf{\upshape z}\in\mathcal{Z}}\exp\left(-\frac{||\textbf{\upshape Gz}||_2^2+s\sqrt{\mu}\cdot||\textbf{\upshape Gz}||_1}{\mu}\right)\geq 2^\Delta\right]-\text{\upshape\sffamily neg}(\kappa)\\
\geq & 1-\text{\upshape\sffamily neg}(\kappa).
\end{align*}
\end{proof}

\noindent We put the previous results together in order to show the main theorem.\\

\noindent\textit{Proof of Theorem} \ref{lweskthm}. By Theorem \ref{wtacr} the $\text{\upshape LWE}(\kappa,\lambda,q,D_\nu)$ problem is hard for $\nu=(\alpha q)^2/(2\pi)>2\kappa/\pi$ if there exists no efficient quantum algorithm approximating the decision-version of the shortest vector problem ($\text{\upshape GAPSVP}$) and the shortest independent vectors problem ($\text{\upshape SIVP}$) to within $\tilde{O}(\kappa/\alpha)$ in the worst case. Let $q=q(\kappa)=\text{\upshape\sffamily poly}(\kappa)$, $s=s(\kappa)=\omega(\log(\kappa))$ and $\lambda>3\kappa$. Then for $\Delta=\omega(\log(\kappa))$, Lemma \ref{nonloss} (setting $r=s\sqrt{\mu}$), Lemma \ref{lossypr} and Lemma \ref{lossysklem} provide that Construction \ref{lossysk} gives us a family of $\Delta$-lossy codes for the symmetric Skellam distribution with variance $\mu\geq\lambda^2\nu$. By Theorem \ref{lossythm} this is sufficient for the hardness of the $\text{\upshape LWE}(\kappa,\lambda,q,\text{\upshape Sk}_\mu)$ problem. Setting $\rho=\alpha\lambda$ yields $(\rho q)^2>4\lambda^2\kappa$ and the claim follows. 
\hfill$\Box$\\

\noindent By Theorem \ref{lwestod} we get the hardness of the DLWE problem as a corollary. 

\begin{Cor}\label{dlweskcor}
Let $\kappa$ be a security parameter 
and let $\lambda=\lambda(\kappa)=\text{\upshape\sffamily poly}(\kappa)$ with $\lambda>3\kappa$. Let $q=q(\kappa)=\text{\upshape\sffamily poly}(\kappa)$ be a sufficiently large prime modulus and $\rho>0$ such that $\rho q\geq 2\lambda\sqrt{\kappa}$. If there exists a PPT-algorithm that solves the $\text{\upshape DLWE}(\kappa,\lambda,q,{\text{\upshape Sk}_{(\rho q)^2/4}})$ problem with non-negligible probability, 
then there exists an efficient quantum-algorithm that approximates the decision-version of the shortest vector problem ($\text{\upshape GAPSVP}$) and the shortest independent vectors problem ($\text{\upshape SIVP}$) to within $\tilde{O}(\lambda\kappa/\rho)$ in the worst case.
\end{Cor}

\section{Application to Differential Privacy}\label{applic}

We turn to showing how the previous result contributes to building prospective post-quantum secure protocols for differential privacy with a relatively high accuracy. In contrast to the $\text{\upshape LWE}(\kappa,\lambda,q,D_\nu)$ problem, note that for the hardness of the $\text{\upshape LWE}(\kappa,\lambda,q,\text{\upshape Sk}_\mu)$ problem we need the standard deviation $\sqrt{\mu}$ of the symmetric Skellam distribution to grow linearly in the number $\lambda$ of equations.

\subsection{Security}\label{secsubsec}

As mentioned in the introduction, the notion of Private Stream Aggregation (PSA) was introduced in \cite{2} and in \cite{40} it was shown that a PSA scheme can be built upon any key-homomorphic weak pseudo-random function and some security guarantees were provided. In the next theorem we recap the result from \cite{40} in a brief form.

\begin{Thm}[Weak PRF gives secure protocol \cite{40}]\label{PSATHEOREM}
Let $\kappa$ be a security parameter, and $m,n\in\mathbb{N}$ with $\log(m)=\text{poly}(\kappa),n=\text{poly}(\kappa)$. Let $(G,\cdot), (S,*)$ be finite groups and $G^\prime\subseteq G$. For some finite set $M$, let \[\mathcal{F}=\{\text{\upshape\sffamily F}_s\,|\,\text{\upshape\sffamily F}_s:M\to G^\prime\}_{s\in S}\] be a (possibly randomised) \textit{weak PRF family} and let \[\varphi:\{-mn,\ldots,mn\}\to G\] be a mapping. Let the algorithm \textbf{\mbox{\upshape \sffamily Setup}} be defined as follows.
\begin{description}
\item \textbf{\mbox{\upshape\sffamily Setup}}: $(\mbox{\upshape\sffamily pp},T,s_0,s_1,\ldots,s_n)\leftarrow \mbox{\upshape\sffamily Setup}(1^\kappa)$, where $\mbox{\upshape\sffamily pp}$ are parameters of $G,G^\prime,S,M,\mathcal{F},\varphi$. The keys are $s_i\leftarrow\mathcal{U}(S)$ for all $i\in[n]$ with $s_0=(\bigast_{i=1}^n s_i)^{-1}$ and $T\subset M$ such that all $t\in T$ are chosen uniformly at random from $M$.
\end{description}
Then for any ppt algorithm  in the non-adaptive compromise model 
the following algorithm generates ciphers indistinguishable under a chosen plaintext attack:
\begin{description}
\item\textbf{\mbox{\upshape\sffamily PSAEnc}}: $c^{(t)}_i\leftarrow \text{\upshape\sffamily F}_{s_i}(t)\cdot \varphi(x^{(t)}_i)$ for $x^{(t)}_i\in\{-m,\ldots,m\}, t\in T$.\,\,\,\,\,\,\,\,\,\,\,\,\,\,\,\,\,\,\,\,\,\,\,\,\,\,\,\,\,\,\,\,\,\,\,\,\,\,\,\,\,\,\,\,\,\,\,\,\,\,\,\,\,\,\,\,\,\,\,\,\,\,\,\,\,\,\,\,\,\,\,\,\,\,\,\,\,\,
\end{description}
Moreover, if $\mathcal{F}$ contains only deterministic functions that are \textit{homomorphic over $S$} 
and if $\varphi$ is an $mn$-isomorphic embedding, then the following algorithm correctly decrypts $\sum_{i=1}^n x^{(t)}_i$ for all $t$:
\begin{description}
\item\textbf{\mbox{\upshape\sffamily PSADec}}: compute $\varphi\left(\sum_{i=1}^n x^{(t)}_i\right)=\text{\upshape\sffamily F}_{s_0}(t)\cdot c^{(t)}_1\cdot\ldots\cdot c^{(t)}_n$ and invert.\,\,\,\,\,\,\,\,\,\,\,\,\,\,\,\,\,\,\,\,\,\,\,\,\,\,\,\,\,\,\,\,\,\,\,\,\,\,\,\,\,\,\,\,\,\,\,\,\,\,\,\,\,\,\,\,\,\,\,\,\,\,\,\,\,\,\,\,\,\,\,\,\,\,\,\,\,\,\,\,\,\,\,\,\,\,\,\,\,\,\,
\end{description}
\end{Thm}

We can build an instantiation of Theorem \ref{PSATHEOREM} (without correct decryption) based on the $\text{\upshape DLWE}(\kappa,\lambda,q,\chi)$ problem as follows. Set $S=M=\mathbb{Z}_q^\kappa, G=\mathbb{Z}_q$, choose $s_i\leftarrow \mathcal{U}(S)$ for all $i=1,\ldots,n$ and $s_0=-\sum_{i=1}^n s_i$, set $\text{\upshape\sffamily F}_{s_i}(t)=\langle t,s_i\rangle +e^{(t)}_i$ (which is a so-called \textit{randomised} weak pseudo-random function as described in \cite{49} and \cite{44}), where $e^{(t)}_i\leftarrow\chi$ and let $\varphi$ be the identity function. Here the decryption function is defined by \[\text{\upshape PSADec}_{s_0}(c^{(t)}_1,\ldots,c^{(t)}_n)=\langle t,s_0\rangle+\sum_{i=1}^n c^{(t)}_i=\sum_{i=1}^n x^{(t)}_i+\sum_{i=1}^n e^{(t)}_i.\]
Thus, the decryption is not perfectly correct any more, but randomised and it gives us a noisy sum. 

\begin{Exm}\label{discretegaussianexample}
Let $\chi=D_{\nu/n}$ with variance $\nu/n=2\kappa/\pi$, then the $\text{\upshape DLWE}(\kappa,\lambda,q,\chi)$ problem is hard due to Theorem $\ref{wtacr}$ and Theorem $\ref{lwestod}$. Thus, the above scheme is secure in the sense of Theorem $\ref{PSATHEOREM}$.
\end{Exm}

\begin{Exm}\label{skellamexample}
Let $\chi=\text{\upshape Sk}_{\mu/n}$ with variance $\mu/n=\lambda^2\kappa$ (where $\lambda=\lambda(\kappa)=\text{\upshape\sffamily poly}(\kappa)$ with $\lambda>3\kappa$), then the $\text{\upshape DLWE}(\kappa,\lambda,q,\chi)$ problem is hard due to Corollary $\ref{dlweskcor}$ and the above scheme is secure in the sense of Theorem $\ref{PSATHEOREM}$.
\end{Exm}

\subsection{Privacy}\label{privsubsec}

Moreover, the total noise $\sum_{i=1}^n e^{(t)}_i$ in Example \ref{skellamexample} is distributed according to $\text{\upshape Sk}_\mu$ due to Lemma \ref{sksum}. Thus, in contrast to the total noise in Example \ref{discretegaussianexample}, the total noise in Example \ref{skellamexample} preserves the distribution of the single noise and can be used for preserving differential privacy of the correct sum by splitting the task of perturbation among the users.\\
We provide the definition of $(\epsilon,\delta)$-differential privacy (\mbox{\upshape\sffamily DP}) and a bound on the variance $\mu$ of the symmetric Skellam distribution that is needed in order to preserve $(\epsilon,\delta)$-\mbox{\upshape\sffamily DP}.\\
We recall that a randomised mechanism preserves differential privacy if its application on two adjacent databases, i.e. databases which differ in one entry only, leads to close distributions of the output.

\begin{Def}[Differential Privacy~\cite{8}]
Let $\mathcal{R}$ be a (possibly infinite) set and let $n\in\mathbb{N}$. A randomised mechanism $\mathcal{A}:\mathcal{D}^n\to\mathcal{R}$ preserves $(\epsilon,\delta)$-differential privacy, if for all adjacent databases $D_0, D_1\in\mathcal{D}^n$ and all $R\subseteq\mathcal{R}$:
\[\Pr[\mathcal{A}(D_0)\in R]\leq e^\epsilon\cdot \Pr[\mathcal{A}(D_1)\in R]+\delta.\]
The probability space is defined over the randomness of $\mathcal{A}$.
\end{Def}

Typically $(\epsilon,\delta)$-\mbox{\upshape\sffamily DP} is achieved by properly perturbing the correct statistics. The next theorem shows how to use the Skellam distribution for this task.

\begin{Thm}[Skellam mechanism preserves \mbox{\upshape\sffamily DP} \cite{40}]\label{skmech} Let $\epsilon>0$ and let $0<\delta<1$. For all databases $D\in\mathcal{D}^n$ the randomised mechanism \[\mathcal{A}_{Sk}(D):=f(D)+Y\]
preserves $(\epsilon,\delta)$-\mbox{\upshape\sffamily DP} with respect to any query $f$ of sensitivity $S(f)$, where $Y\leftarrow\text{\upshape Sk}_\mu$ with 
\[\mu\geq\frac{\log(1/\delta)}{1-\cosh(\epsilon/S(f))+(\epsilon/S(f))\cdot\sinh(\epsilon/S(f))}.\]
\end{Thm}

\begin{Rem}\label{skdprem} The bound on $\mu$ from Theorem \text{\upshape \ref{skmech}} is smaller than $2\cdot (S(f)/\epsilon)^2\cdot\log(1/\delta)$, thus the standard deviation $\mu$ of $Y\leftarrow\text{\upshape Sk}_\mu$ may be assumed to be linear in $S(f)/\epsilon$ (for constant $\delta$).
\end{Rem}

Suppose that adding symmetric Skellam noise with variance $\mu$ preserves $(\epsilon,\delta)$-\mbox{\upshape\sffamily DP}. We define $\mu_{user}=\mu/n$. Since the Skellam distribution is reproducible, the noise addition can be executed in a distributed manner: each (non-compromised) user simply adds (independent) symmetric Skellam noise with variance $\mu_{user}$ to her own value in order to preserve the privacy of the final output.

\subsection{Accuracy}\label{acsubsec}

\begin{Thm}[Accuracy of the Skellam mechanism \cite{40}]\label{erroranalysis} Let $\epsilon>0$ and $0<\delta<1$. Then for all $0<\beta<1$ the mechanism specified in Theorem $\ref{skmech}$ has $(\alpha,\beta)$-accuracy, where
\[\alpha=\frac{S(f)}{\epsilon}\cdot\left(\log\left(\frac{2}{\beta}\right)+\log\left(\frac{1}{\delta}\right)\right).\]
\end{Thm}

This means, the error does not exceed $\alpha$ with probability at least $1-\beta$. Theorem \ref{PSATHEOREM} indicates that the set $T$ contains all the time-frames where a query can be executed. For simplicity we assume that all queries are independent, i.e. the arguments of all queries are independent. As pointed out in section \ref{lwesubsec} we identify the number of queries with the number of equations in the instance of the $\lambda$-bounded LWE problem, thus $|T|=\lambda$. (A result in \cite{47} indicates that for an efficient and accurate mechanism this number cannot be substantially larger than $n^2$, where $n$ is the number of users in the network.) Due to sequential composition (see for instance Theorem $3$ of \cite{48}), in order to preserve $(\epsilon,\delta)$-\mbox{\upshape\sffamily DP} in all $\lambda$ queries together, the executed mechanism must preserve $(\epsilon/\lambda,\delta)$-\mbox{\upshape\sffamily DP} in every single query. Therefore the following holds: suppose $\text{\upshape Sk}_{\mu^\prime}$-noise is sufficient in order to preserve $(\epsilon,\delta)$-\mbox{\upshape\sffamily DP} in a single query. Then, due to Remark \ref{skdprem}, we must use $\text{\upshape Sk}_{\lambda^2\mu^\prime}$-noise in order to preserve $(\epsilon,\delta)$-\mbox{\upshape\sffamily DP} in all $\lambda$ queries.\\
With Theorem \ref{erroranalysis} we obtain $(\alpha,\beta)$-accuracy for every single query executed during $T$ with
\[\alpha=\frac{S(f)}{\epsilon/\lambda}\cdot\left(\log\left(\frac{2}{\beta}\right)+\log\left(\frac{1}{\delta}\right)\right)=O\left(\frac{S(f)\lambda}{\epsilon}\right),\]
which is optimal with respect to sequential composition.

\subsection{Combining Security, Privacy and Accuracy}

Set $S(f)=m$ and $\mu=2\cdot (m\lambda/\epsilon)^2\cdot\log(1/\delta)$. From the discussion from above it follows that if every user adds $\text{\upshape Sk}_{\mu/n}$-noise to her data in every time-step $t\in T$, then this is sufficient in order to preserve $(\epsilon,\delta)$-\mbox{\upshape\sffamily DP} in all $\lambda$ sum-queries that are executed during $T$, where for each time-step $t\in T$ the data of each user comes from $\{-m,\ldots,m\}$.\\ 
Furthermore, if for a security parameter $\kappa$ we have that $\mu/n=\lambda^2\kappa$, then we obtain a secure protocol for analysing sum-queries, where the security is based on prospectively hard lattice problems. As shown in \cite{40}, a combination of these two results provides \textit{computational} $(\epsilon,\delta)$-\mbox{\upshape\sffamily DP} in all $\lambda$ sum-queries. Assume that we want to find values for $\epsilon, \delta$ such that when every user adds $\text{\upshape Sk}_{\mu/n}$-noise to her data with $\mu=2\cdot (m\lambda/\epsilon)^2\cdot\log(1/\delta)$ to preserve $(\epsilon,\delta)$-\mbox{\upshape\sffamily DP} of the final statistics, then the same noise suffices for providing security. Therefore the following must be satisfied:
\[2\cdot (m\lambda/\epsilon)^2\cdot\log(1/\delta)\geq n\lambda^2\kappa.\]
This inequality holds for 
\[\epsilon=\epsilon(\kappa)\leq\sqrt{\frac{2 m^2\cdot\log(1/\delta)}{\kappa\cdot n}},\numberthis\label{epsilonupperbound}\]
indicating that $\epsilon=\epsilon(\kappa)$ depends on $1/\kappa$. Note that this is consistent with the original definition of computational \mbox{\upshape\sffamily DP} in \cite{15}. Thus, in addition to a privacy/accuracy trade-off we also get a security/accuracy trade-off. More specifically, depending on $\kappa$ and $n$ we obtain a tight lower bound on the $(\alpha,\beta)$-accuracy for every single query executed during $T$:
\begin{align*}
\alpha= & \frac{m}{\epsilon/\lambda}\cdot\left(\log\left(\frac{2}{\beta}\right)+\log\left(\frac{1}{\delta}\right)\right)\\
\geq & \lambda\cdot\sqrt{\frac{\kappa\cdot n}{2\cdot\log(1/\delta)}}\cdot\left(\log\left(\frac{2}{\beta}\right)+\log\left(\frac{1}{\delta}\right)\right)\\
= & \Omega(\lambda\sqrt{\kappa\cdot n}).
\end{align*}

\begin{Exm}
There are $n=20,000$ users in a network with variable time-series data falling in an interval of $\{-1000,\ldots,1000\}$, i.e. $m=1000$. Note that the aggregated sum can not exceed $20,000,000$. The users want to preserve computational $(\epsilon,\delta)$-\mbox{\upshape\sffamily DP} with $\epsilon=1$ and $\delta=0.1$ while evaluating all $\lambda$ sum-queries over a time period $T$. They use the DLWE-based secure protocol for communicating with an untrusted aggregator and generate symmetric Skellam noise. For a hard DLWE problem, the security parameter is chosen to be $\kappa=200$. For these parameters, inequality $\ref{epsilonupperbound}$ is satisfied. For the $(\alpha,\beta)$-accuracy per query it holds that \[\alpha=\frac{m}{\epsilon/\lambda}\cdot\left(\log\left(\frac{2}{\beta}\right)+\log\left(\frac{1}{\delta}\right)\right)=1000\lambda\cdot\left(\log\left(\frac{2}{\beta}\right)+\log(10)\right).\]
\end{Exm}


\section{Conclusions}

In this work we provided a worst-to-average-case connection from conjecturally hard lattice problems to the LWE problem with errors following a symmetric Skellam distribution. Our proof relies on the notion of lossy codes from \cite{39}. An implication of this result is the construction of the first prospective post-quantum Private Stream Aggregation scheme for data analyses under differential privacy where the errors are used both for security of the scheme and for the distributed noise generation for preserving differential privacy. An interesting further direction is to reduce the size of the variance that is necessary for the hardness of the LWE problem with errors following a symmetric Skellam distribution, especially to abolish the dependence on the number of LWE-samples. Another problem to face is to show the hardness of the Ring LWE problem (a more efficient version of LWE introduced in \cite{45}) with errors following a symmetric Skellam distribution and to establish a corresponding search-to-decision reduction. Sufficient conditions on the error distribution for the existence of a search-to-decision reduction for the Ring LWE problem were provided in \cite{55}. The Skellam distribution does not seem to satisfy these conditions. Thus, we require a different proof than in \cite{55}.

\addcontentsline{toc}{chapter}{Literatur}
\bibliography{Literatur}

\end{document}